\documentclass[a4paper,11pt]{article}
\pdfoutput=1
\usepackage{amsmath,amsthm, amssymb, graphicx, float}
\usepackage{caption}
\usepackage[font=small,labelfont=normalfont]{subcaption}
\usepackage[a4paper,hmargin=1.2in,vmargin=1.2in]{geometry} 
\usepackage{url}
\usepackage{enumitem}
\usepackage{xcolor}
\usepackage[pdfpagelabels, naturalnames,colorlinks]{hyperref}

\hypersetup{colorlinks,breaklinks,
            urlcolor=[rgb]{.8,0.4,.8},
            linkcolor=[rgb]{0.2,0.2,.8},
            citecolor=[rgb]{.8,.2,.2}}

\newtheorem{theorem}{Theorem}
\newtheorem{lemma}{Lemma}

\newtheorem{observation}{Observation}

\graphicspath{{}{figures/}}

\def\R{\mathbb R}

\newcommand{\LZ}{\it{L}}
\newcommand{\LMZ}{\it{LM}}
\newcommand{\RMZ}{\mathit{RM}}
\newcommand{\RZ}{\mathit{R}}
\newcommand{\bl}{\mathit{b}}
\newcommand{\rl}{\mathit{r}}

\begin{document}

\title
{Saturated simple and 2-simple topological graphs with few edges}

\author{P\'eter Hajnal\\
            Bolyai Institute, University of Szeged,\\
            Aradi v\'ertan\'uk tere 1, Szeged, Hungary 6720\\
\smallskip            \\
        Alexander Igamberdiev\\
            Institut f\"ur Mathematische Logik und Grundlagenforschung\\
            Fachbereich Mathematik und Informatik der Universit\"at M\"unster\\
            Einsteinstra{\ss}e 62, D-48149 M\"unster, Germany\\
\smallskip            \\
        G\"unter Rote\\
            Freie Universit\"at Berlin,
            Institut f\"ur Informatik,\\
            Takustra{\ss}e 9, D-14195 Berlin, Germany\\
\smallskip            \\
        Andr\'e Schulz\\
            Institut f\"ur Mathematische Logik und Grundlagenforschung\\
            Fachbereich Mathematik und Informatik der Universit\"at M\"unster\\
            Einsteinstra{\ss}e 62, D-48149 M\"unster, Germany}


\maketitle




\begin{abstract}
A \emph{simple topological graph} is a topological graph
in which any two edges have at most one common point, which
is either their common endpoint or a proper crossing.
More generally, in a \emph{$k$-simple topological graph},
every pair of edges has at most $k$ common points of this kind.
We construct \emph{saturated} simple and 2-simple graphs
with few edges. These are $k$-simple graphs
in which no further edge can be added. 
We improve the previous upper bounds
of Kyn\v{c}l, Pach, Radoi\v{c}i\'c, and T\'oth~\cite{kprt} and show that
there are saturated simple graphs on~$n$ vertices with only $7n$~edges and
 saturated 2-simple graphs on~$n$ vertices with $14.5n$~edges. 
As a consequence, $14.5n$~edges is also a new upper bound for 
$k$-simple graphs (considering all values of $k$).
We also construct saturated simple and 2-simple graphs that have some vertices
with low degree.
\end{abstract}


\section{Introduction}\label{intro}

Let $G=(V,E)$ be a finite simple graph.
A  \emph{drawing} of $G$ is a map $\delta\colon  V\cup E \to \R^2$
that is one-to-one on 
$\delta|_V\colon V\to\R^2$, i.e.,~$\delta$ assigns the vertices of the graph to different points of the plane. 
Furthermore, we require that $\delta|_E\colon E\to\mathcal C$,
where $\mathcal{C}$ is a set of ``nice'' non-self-intersecting
curves with two boundary points of the plane. For example we might think of 
$\mathcal C$ as the set of all Jordan curves or, more elementary, 
of the set of all simple polygonal curves.
For simplicity, we will not distinguish between an edge and
the curve on which it is embedded,
and between a vertex and 
the point on which it is embedded.
We assume that 
for any $e=xy\in E$ the  edge $\delta(e)$ is a curve
connecting $\delta(x)$ and $\delta(y)$ and it doesn't
go through any other vertex, and also that 
any two different edges meet at finitely many
points and any meeting point --- that is not a common
endvertex --- is a proper crossing of the two curves.

The pair $(G,\delta)$, i.e., a graph with a drawing, is called a \emph{topological
graph}. A topological graph $(G,\delta)$ is \emph{simple} if in $\delta$ two edges have at most
one common point. More generally, the topological graph is called $k$-simple if in $\delta$
two edges have at most $k$ common points. For both simple and $k$-simple graphs we 
do not allow self-intersecting edges. 
A topological graph is a \emph{geometric graph} if all its
edges are drawn as straight-line segments.
Obviously, every geometric graph is simple, 
provided that the vertices are placed
in general position.
Thus, every graph 
 has simple drawings.

For a graph property~$T$, a graph $G$ is \emph{$\mathcal T$-saturated} if $G$
has property $\mathcal T$, but the addition of 
any edge joining two non-adjacent vertices of
$G$ violates property $\mathcal T$.
Often structures with property $\mathcal T$ 
are quite hard to grasp, but
$\mathcal T$-saturated structures
might have a more useful character.
We direct the interested reader to applications of the
saturation technique~\cite{ehm,kt,l}.
This notion can be naturally extended to hypergraphs.
A thorough survey by Faudree, Faudree, and Schmitt~\cite{ffs}
discusses the case when
property $\mathcal T$ is ``not having $F$ as a sub(hyper)graph''.

In this paper we study \emph{saturated} $k$-simple topological graphs. 
These are topological graphs that are $k$-simple, but no edge can be added
without violating the $k$-simplicity of the drawing.
Saturated planar drawings are triangulations and have therefore due to Euler's formula
 $3n-6$ edges.
Recently, Kyn\v{c}l, Pach, Radoi\v{c}i\'c, and G.~T\'oth~\cite{kprt}  started to investigate
saturated simple $k$-simple graphs.
The maximum number of edges a saturated simple topological graph can have is clearly $\binom{n}{2}$, since 
the geometric graph of $K_n$ with vertices in general position is a simple drawing. 
The more intriguing questions ask about the minimum number of edges for
saturated $k$-simple topological graph.
One of the main results of Kyn\v{c}l et al.~\cite{kprt} is a
construction of sparse saturated simple and $k$-simple topological
graphs. We denote by $s_k(n)$ the minimum number of edges a saturated $k$-simple graph with~$n$ vertices can have.
Their upper bound on $s_k(n)$ 
is a linear function of $n$, for $n$ being the number of vertices; see Table~\ref{tab:results} for the bounds
obtained by Kyn\v{c}l et al.~\cite{kprt}.
The gap between the best known upper and lower bounds for $s_k(n)$ is quite substantial. We only know that $s_1(n)\ge 1.5 n$ and that
$s_k(n)\ge n$~\cite{kprt}.

\setlength{\tabcolsep}{.4em}
\begin{table}[b]
\centering
\begin{tabular}{l|c|c|c|c|c|c|c|c}
 $k$ &  1& 2 & 3 & 4 & 5 & 6,8,10 & 7 &  9, $\ge11$  \\
 \hline
 old upper bounds~\cite{kprt}& $17.5n$   & $16n$  & $14.5n$ &$13.5n$  & $13n $ & $9.5n$  & $10n$ & $7n$     \\
 new upper bounds & $\mathbf{7n}$ & $\mathbf{14.5n}$ &  &  &  &  &  &  \\
\end{tabular}
\vskip4mm
\caption{Old and new upper bounds for~$s_k(n)$, the minimum number of edges in a saturated $k$-simple graph with~$n$ vertices.}
\label{tab:results}
\end{table}

\paragraph{Our contribution.}
We improve the upper bounds for $s_k(n)$ for $k=1,2$. We do this by
showing that for any positive integer $n$ there exists a 
saturated simple topological graph with at most $7n$ edges~(in Sect.~\ref{sec:simple}), and a 
saturated 2-simple graph with at most~$14.5n$ edges~(in~Sect.~\ref{sec:2simple}).
Sections~\ref{sec:simple} and~\ref{sec:2simple} are independent.
This result
also implies that there are saturated $k$-simple graphs with at most~$14.5n$ edges for every~$k$.
See also Table~\ref{tab:results} for a comparison with the old bounds.
Our proofs are constructive, i.e., we can explicitly present the sparse saturated graphs. 
%
%

We complete our results by studying \emph{local saturation} of topological graphs.
Here, local saturation refers to drawings in which one (or several) vertices have a small
vertex degree even though the full drawing might not be the sparsest. Such observations might be 
helpful in further studies, e.g.,~if we want to estimate techniques for proving lower bounds
that are based on the minimum vertex degree in saturated graphs. We show that there are 
saturated simple graphs that have a vertex of degree~4, and saturated simple graphs in which 10~percent
of the vertices have degree~5. For saturated 2-simple graphs we can prove that there are 
drawings with minimum degree~12. The currents lower bounds for $s_k(n)$ are obtained by bounding the minimum 
vertex degree in saturated $k$-simple graphs~\cite{kprt}. Our results show the limits of this approach. 


\section{Saturated simple topological graph with few edges}
\label{sec:simple}
In this section we give a construction that generates 
sparse saturated simple graphs.
We start with 
defining a graph~$G$, parametrized by an integer~$k$, 
with~$n=6k$ vertices and~$9k-6$ edges. This 
graph is the backbone of our sparse saturated graph.

The drawing is best visualized on the surface of a long circular
cylinder.
Fig.~\ref{cylinder} shows an unrolling of the cylinder into the
plane.
The cylinder is obtained by cutting the drawing along the two dotted
lines and gluing the top and the bottom together.
The vertices of the graph are placed in a $3\times 2k$-grid-like fashion. 
We draw the vertices together in pairs, with each vertex
$X_i^L$ on the \emph{left} and the corresponding vertex
$X_i^R$ on the \emph{right}, for  $X=A,B,C$ and $i=1,\ldots ,k$.
We refer to the vertices whose label have the subscript~$i$
as the \emph{$i$-th layer}.
 $G$ is the union of
\begin{itemize}[itemsep=0ex]
\item three vertex-disjoint paths of  \emph{blue edges} connecting 
$A^L_1A^L_2\ldots A^L_k$, $B^L_1B^L_2\ldots B^L_k$,
and $C^L_1C^L_2\ldots C^L_k$,
\item three vertex-disjoint paths of  \emph{red edges} connecting 
$A^R_1A^R_2\ldots A^R_k$, $B^R_1B^R_2\ldots B^R_k$, and
$C^R_1C^R_2\ldots C^R_k$, and
\item $k$ disjoint cycles of \emph{green edges} connecting
$A^L_iB^L_iC^L_i$.
\end{itemize}

\begin{figure}[htbp]
  \centering
  \includegraphics[scale=0.65]{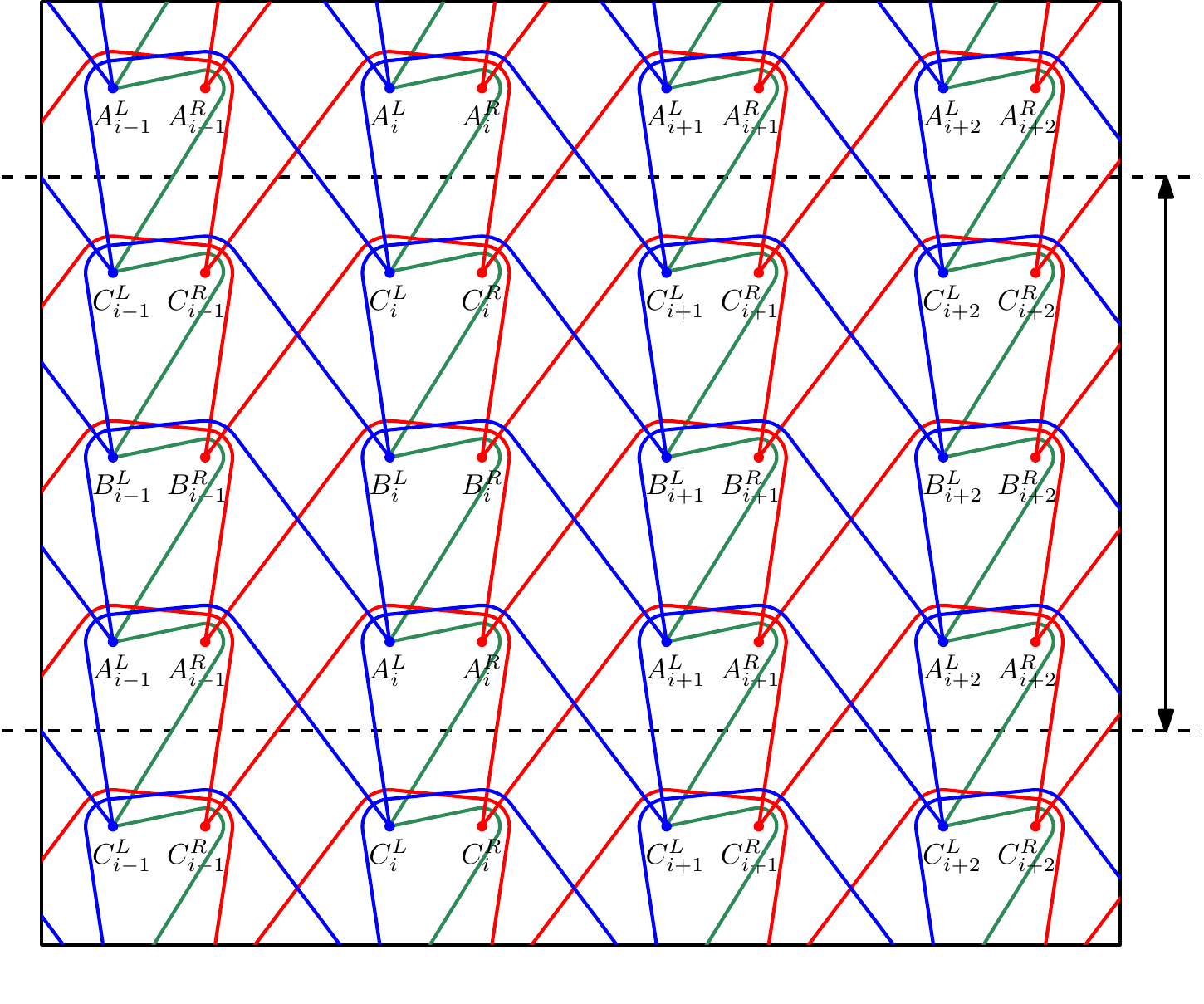}
  \caption{The graph $G$ on an unrolled cylinder.}
  \label{cylinder}
\end{figure}

The cylinder can be homeomorphically mapped into the plane, as shown
in 
Fig.~\ref{circular}
for the red and blue edges only.  The
horizontal directions turn into radial directions. But the resulting
drawings suffer from large distortions, and the left-right symmetry is
lost.  We therefore prefer the cylindrical drawings,
and we extend the cylinder surface
periodically beyond the dotted lines (using the plane as a universal
cover of the cylinder). One should however be aware that
vertices (and edges) that appear as distinct in the figure may denote
the same vertex, as
indicated by the vertex labels.

\begin{figure}[H]
  \centering
  \includegraphics[scale=0.6]{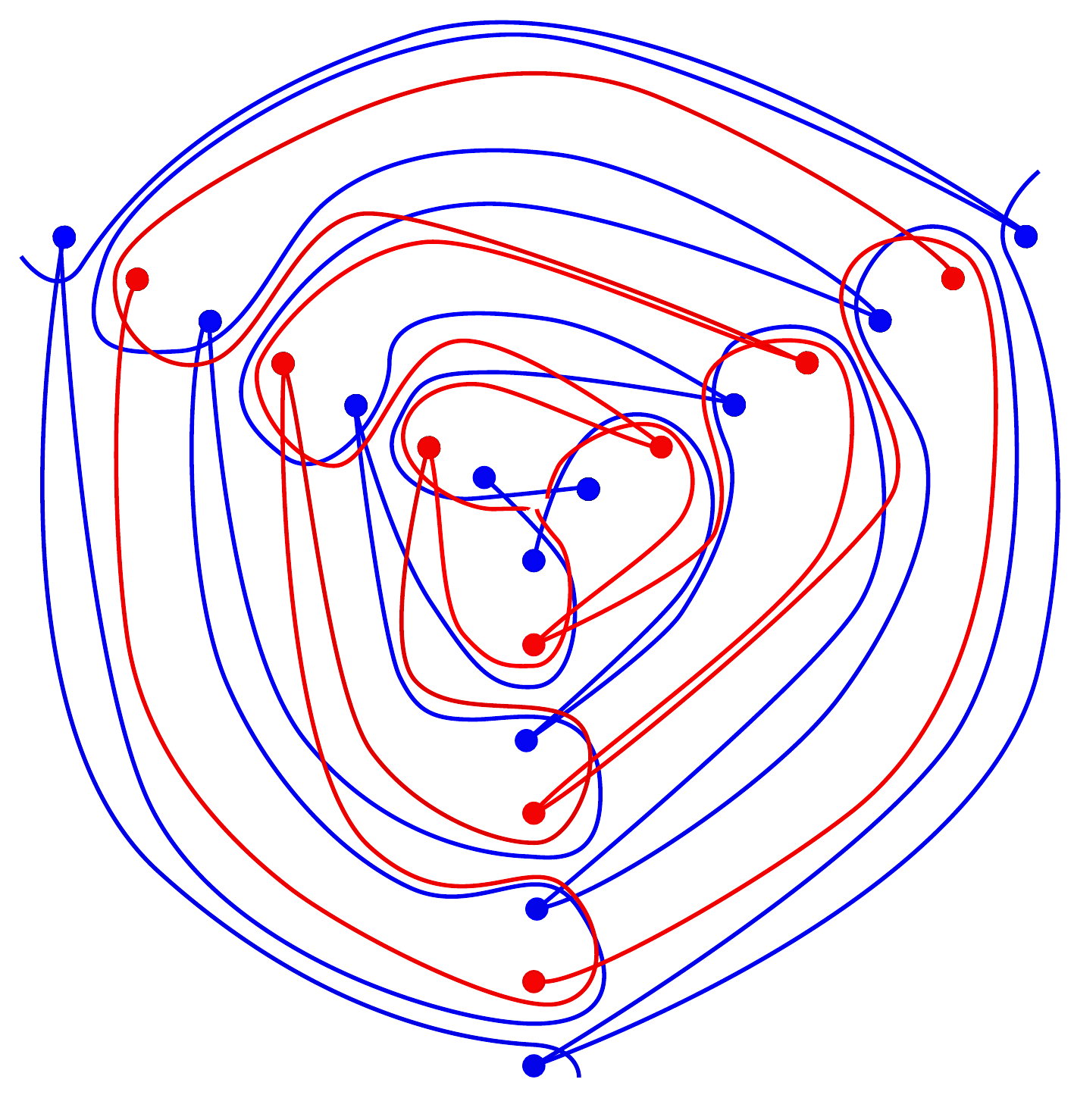}
  \caption{The graph $G$ on the plane.}
  \label{circular}
\end{figure}

We will first consider the graph $G_{RB}$ that omits the green edges, 
because this graph
is more symmetric: with the exception of the vertices
$X^{L/R}_1$ and
$X^{L/R}_k$ near the boundary, all vertices look identical.
Apart from these boundary effects,
the drawing has a rotational symmetry, cyclically shifting the labels
$A\to B\to C\to A$,
 a translational symmetry, shifting indices $i$ up or down,
and a mirror symmetry, exchanging left with right and blue with red.
The green edges destroy this mirror symmetry: there are then two
classes of vertices, the blue vertices
$X_i^L$ and the red vertices
$X_i^R$.

\begin{figure}[htbp]
  \centering
\hrule height0pt
\vskip -12mm
  \includegraphics[scale=0.851]{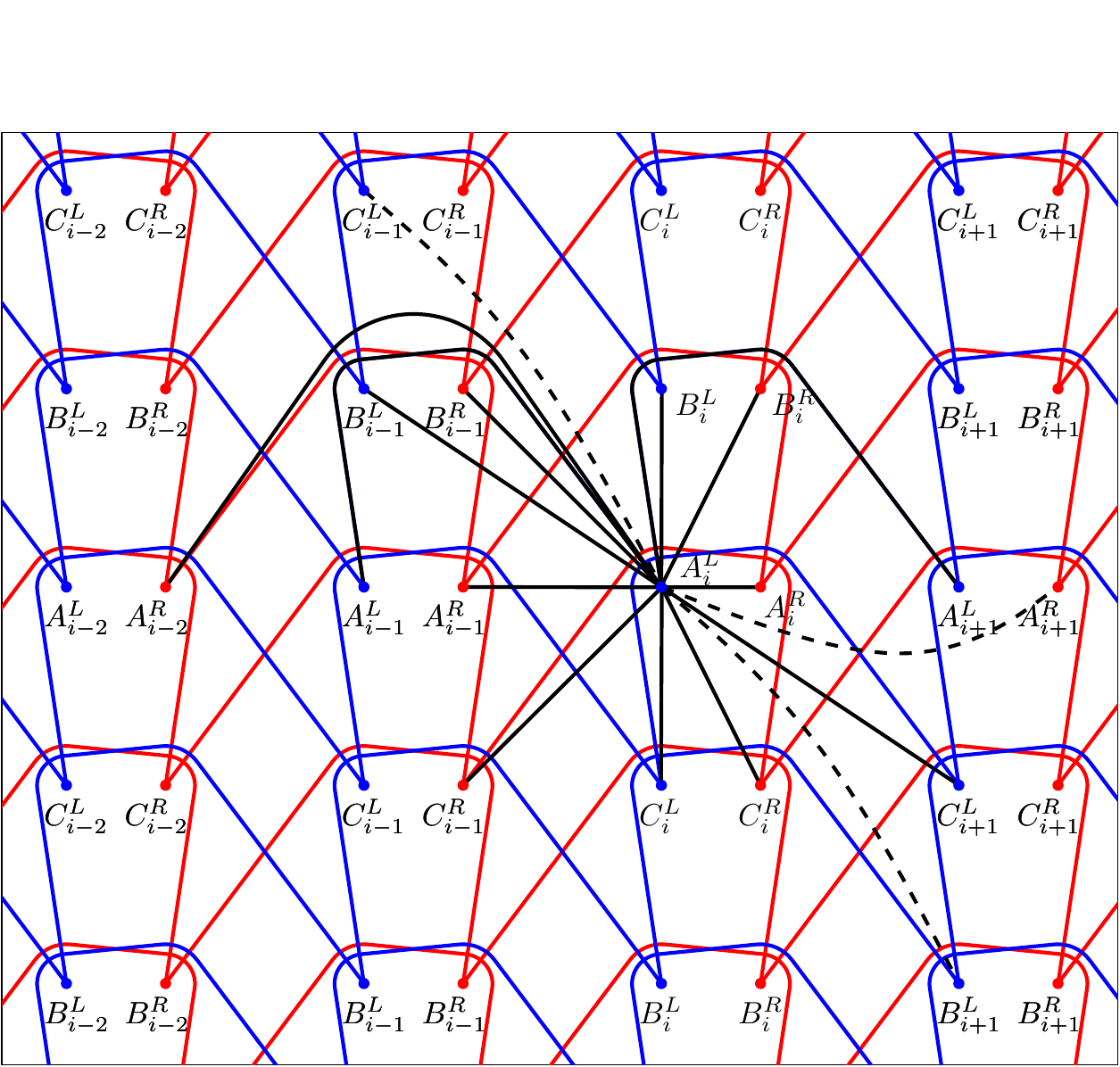}
  \caption{The 16 potential neighbors of a vertex.}
  \label{fig-neighbors}
\end{figure}

Let $G_{RB}$ denote the
topological graph obtained by restricting $G$ to the red and blue edges.
We will show that the maximum degree in any saturated drawing
which extends
$G_{RB}$ is 16. The 16 potential neighbors of a typical vertex
$A_i^L$
are shown in Fig.~\ref{fig-neighbors}.
This establishes that there are saturated drawings with $n$ vertices and 
less than
$8n$ edges.
When the green edges are included,
the three dashed edges in Fig.~\ref{fig-neighbors} become impossible.
Thus, each blue vertex has 13 potential neighbors.
The red vertex $A^R_{i+1}$, which can be taken as a representative of
a typical red vertex, loses $A_l^L$ as a potential neighbor.
Thus, each red vertex has at most 15 potential neighbors.
This improves the upper bound for the smallest number of edges in 
a saturated drawings with $n$ vertices to
$7n$.

\begin{theorem}\label{SimpleTheorem}
Let $s(n)$ denote the minimum number of edges that
a simple saturated drawing with $n$ vertices can have.
Then 
$$s(n)\le 7 n.$$
\end{theorem}

The remainder of this section is devoted to proving the above theorem. We start with the
analysis of the graph $G_{RB}$.


\subsection{The graph $G_{RB}$}

\begin{lemma}\label{neighbors-RB}
  The $16$ potential neighbors of a typical vertex $A^L_i$
in $G_{RB}$ are
all $11$ vertices of levels $i-1$ and $i$
\textup(%
$A^L_{i-1},B^L_{i-1},C^L_{i-1};$
$A^R_{i-1},B^R_{i-1},C^R_{i-1};$
$B^L_{i},C^L_{i};$
$A^R_i,B^R_{i},C^R_{i}$\textup)
plus the $5$ vertices
$A^R_{i-2};$
$A^L_{i+1},B^L_{i+1},C^L_{i+1};$
$A^R_{i+1}$.
\end{lemma}

When any of the neighbors listed above does not exist because $i\le 2$
or $i=k$, the lemma still holds in the sense that the remaining
vertices form the set of potential neighbors.
In the proofs, when we exclude an edge between, say, levels $i$ and $j$, our
arguments will not use edges outside this range.

In the following we will look at the given drawing of $G_{RB}$ (or
$G$) and argue about the additional edges that can be drawn.  The
implicit assumption is that these edges cannot cross any given edge
more than once. Usually, we will regard a new edge as a directed edge,
starting at some vertex and trying to reach another vertex.




A \emph{belt} is a substructure of our drawing. It is formed by
the 12 vertices of two successive layers with their 6 edges between
them, see Fig.~\ref{belt2}.
This drawing 
 separates a large face on the left from a large face on the right.
More precisely, the belt is defined as the part of the plane (or the
cylinder)
which lies between these two large faces (shaded area).


\begin{figure}
\centering
\begin{minipage}{.4\columnwidth}%
    \includegraphics[scale=.70]{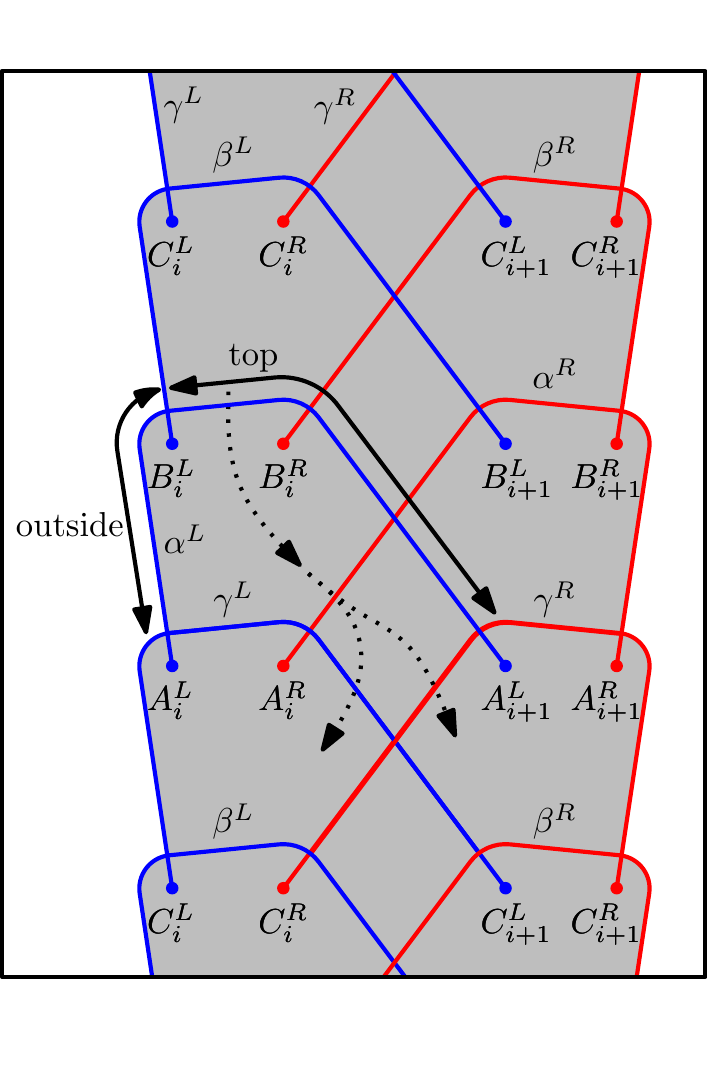}
    \vskip -7mm
\caption{Escape from a belt is difficult
(Lemma~\ref{belt}).}%
\label{belt2}%
\end{minipage}
\qquad
\begin{minipage}{.5\columnwidth}%
    \includegraphics[scale=.65]{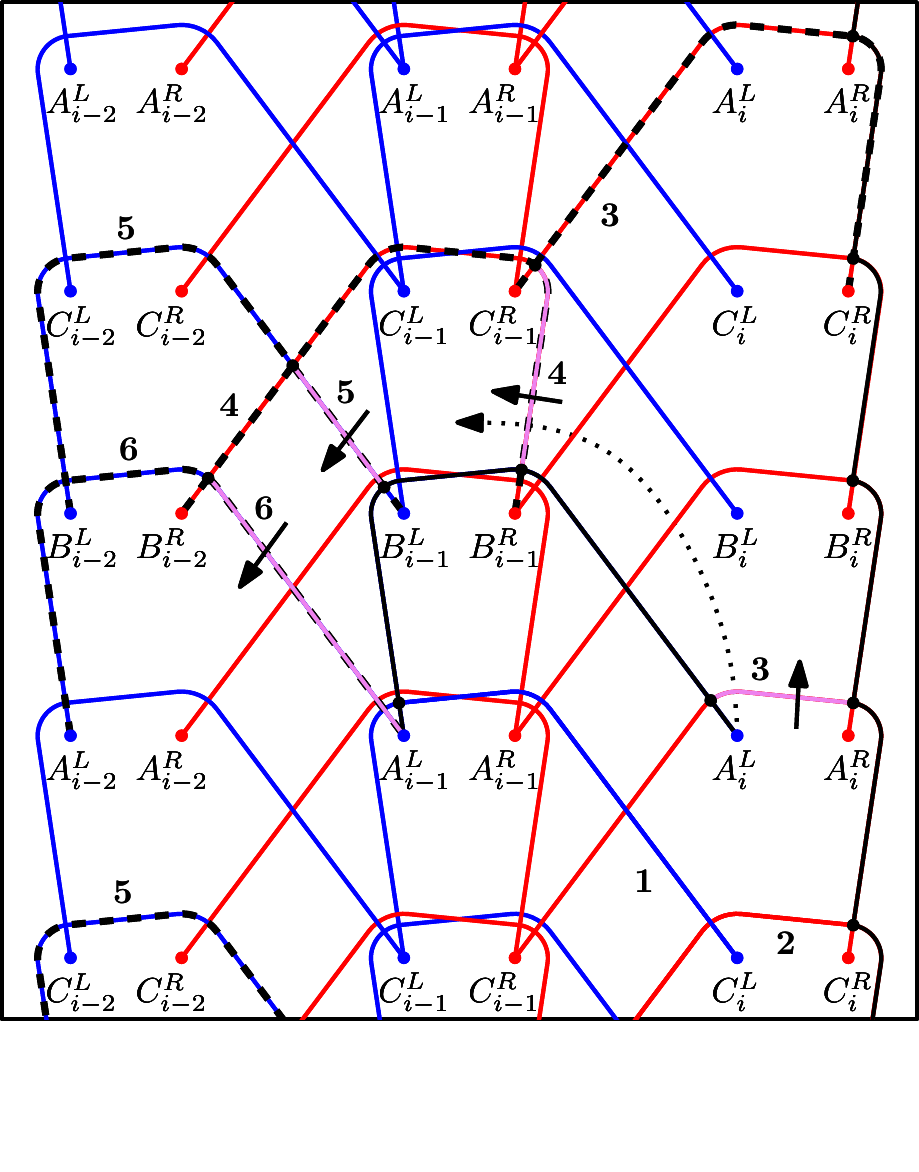}
\vskip -10mm
\caption{The situation discussed in the proof of Lemma~\ref{neighbors-RB} for left side neighbors.}%
\label{left}%
\end{minipage}%
\end{figure}


 We denote the six edges of the belt by
$\alpha^L,
\beta^L,\gamma^L,\alpha^R,\beta^R,
\gamma^R$; as shown in Fig.~\ref{belt2}.
Each edge is cut into six sections by the intersections with
the other edges:
Two sections are little ``stumps'' at the end vertices.
One section belongs to the boundary between the belt and the \emph{outside}.  
The remaining
three sections form the \emph{top part} of the edge.
We say that a new (directed) edge crosses a belt edge
\emph{from the outside} or \emph{from the top} if it crosses the boundary
part or the top part in the appropriate direction.

\begin{lemma}\label{belt}
In a simple drawing that contains
$G_{RB}$, the following holds:
  \begin{enumerate}
\item 
If an edge crosses a belt edge from the top
or from the outside, it must terminate inside the belt.
  \item 
  No edge can cross a belt.
%
  \end{enumerate}
\end{lemma}

\begin{proof}
We start with the following observation:
  If an edge crosses $\alpha_L$
  from the outside or from the top, and 
  it does not terminate at $B_i^L$ or at $B_i^R$, 
  then it must later cross
  $\gamma^L$ or $\gamma^R$ from the top.
This observation holds symmetrically for $\alpha_R$ instead of
$\alpha_L$, and cyclically for the other four belt edges.
Hence, any edge that ``enters'' the belt from the outside
has to continue by crossing another edge from the 
belt from the top. There is no way to leave the belt 
without crossing some edge twice.
\end{proof}

% ...





After these preparations,
we are ready to prove Lemma~\ref{neighbors-RB}.

\noindent\textit{Proof of Lemma~\ref{neighbors-RB}.}
Let us first look at the potential neighbors on the left side.  A
connection from $A^L_i$ to levels $j\le i-3$ is impossible, because
it would have to cross a belt.
For the vertices at level~$i-2$ we observe the following 
(see Fig.~\ref{left} for the edge numbers we are referring to):
When we start from $A^L_i$ we cannot cross the right boundary of the belt
formed by levels $i-1$ and $i$, because then we
would have to cross the whole belt to reach level $i-2$.
If we cross edge~1 or~2 from the top, then, by Lemma~\ref{belt},
we are restricted to the belt defined by level~$i-1$ and~$i$. 
Thus we can regard edge~1 and 2 as
closed from the top. (These edges can later be crossed
from the bottom.) 
We successively conclude that the new edges must cross
the purple parts of the edges 3, 4, 5, and~6.
The endpoints
$B^R_{i-2},B^L_{i-2},A^L_{i-2}$
 of the edges 4, 5, and 6 cannot be taken.
$C^L_{i-2}$ and $C^R_{i-2}$ are enclosed
in a small face delimited by the edges 4, 5, and 6, and cannot be reached.
$A^R_{i-2}$ is thus the only reachable vertex of level $i-2$.

%

Let us turn to the potential neighbors on the right side.  A
connection from $A^L_i$ to levels $j\ge i+3$ is impossible, because
it would have to cross a belt.
Vertices at level $i+2$ cannot be reached either, because 
(i) if we cross the edge forming the left boundary of the belt spanned by the vertices of level~$i$ and~$i+1$
we cannot cross this belt anymore and therefore cannot reach level $i+2$, and (ii) if we cross one 
of the edge in the face that contains $A^L_i$ from the top (edge labeled~1 and~2 in~Fig.~\ref{right-right}),  then, by Lemma~\ref{belt},
we are also restricted to this belt. Thus we are restricted to the shaded region in~Fig.~\ref{right-right}.

\begin{figure}[tb]
  \begin{subfigure}[b]{.48\columnwidth}
    \centering
    \includegraphics[scale=.85]{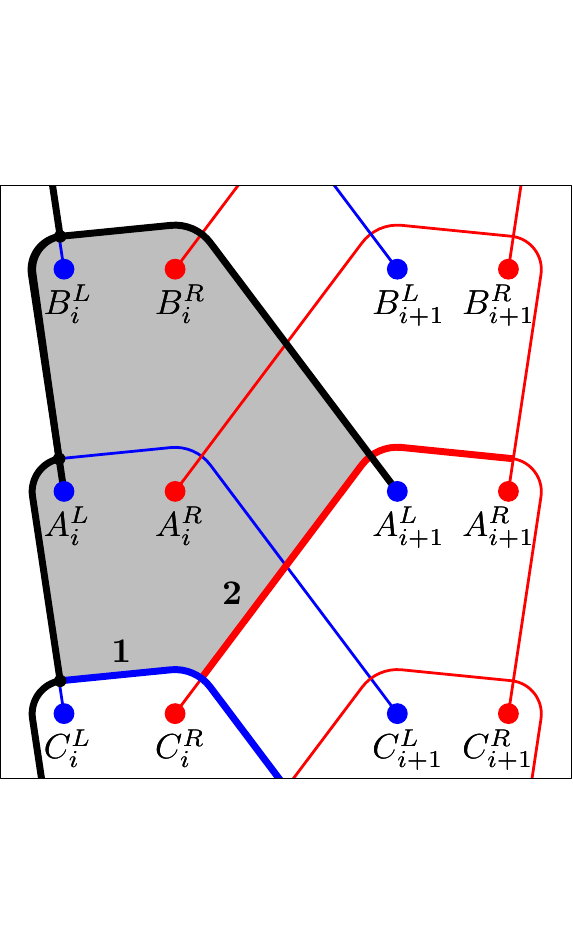}
\vskip -10mm
\hrule height0pt
    \caption{Level $i+2$ cannot be reached from
  $A^L_i$.}
    \label{right-right}
  \end{subfigure}
  \hfill
  \begin{subfigure}[b]{.48\columnwidth}
    \centering
    \includegraphics[scale=.85]{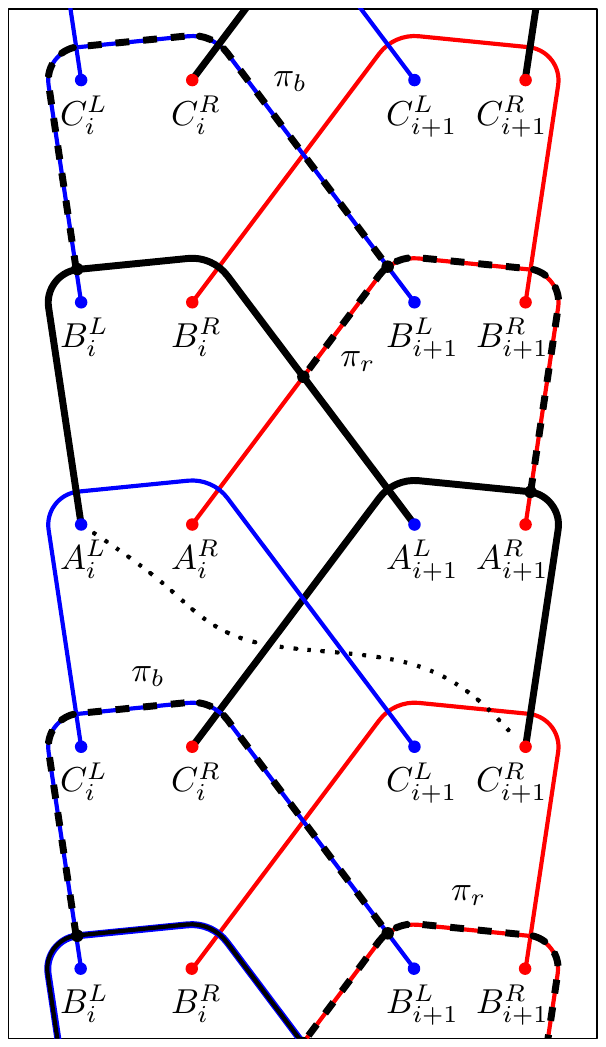}
\vskip -6mm
\hrule height0pt
    \caption{$C^R_{i+1}$ cannot be reached from
  $A^L_i$.}
    \label{right}
  \end{subfigure}
  \caption{Restricting the neighbors to the right.}
\end{figure}
%
%
%
%
%

The vertices $B^R_{i+1}$ and $C^R_{i+1}$ also cannot be neighbors of $A^L_i$.
We discuss the exclusion of $C^R_{i+1}$ as a potential neighbor -- the case 
for $B^R_{i+1}$ is symmetric. 
The edges incident to $A^L_i$ and $C^R_{i+1}$, which we call the \emph{closed edges} cannot be crossed.
The closed edges are depicted as thicker curves in Fig.~\ref{right}.
Consider the portion  of the red edge $\pi_r$ that runs between $A_{i}^R$ and $A_{i+1}^R$ above the closed edges
(see Fig.~\ref{right}).
The curve $\pi_r$ bounds a region below in which the remaining 
edges bounding this region are parts of the closed edges.
Hence, if we enter this region we cannot leave and therefbrachistochroneore we cannot cross $\pi_r$ (see Fig.~\ref{right}).
Let us now consider the partial edge $\pi_b$ that runs between $B_{i+1}^L$ and $B_i^L$
above the closed edges and $\pi_r$. Again, there is a region whose boundary is part
of the closed edges and also $\pi_b$. To enter and leave this region we have to cross either 
one of the closed edges or $\pi_r$, or we have to cross $\pi_b$ twice. Since all these options are
invalid, we have to avoid this region, and therefore are not allowed to cross $\pi_b$.  
We observe that the closed edges
together with $\pi_b$ and $\pi_r$ leave  $A^L_i$ and $C^R_{i+1}$ in different faces, which
shows that these vertices cannot be neighbors unless we cross one edge twice.
\qed

\subsection{The Graph $G$}

Now we turn back to $G$. The additional green edges 
exclude some of the possible edges from
the Lemma~\ref{neighbors-RB}. To analyze the drawing 
of $G$ we have to treat the 
left and right vertices differently.

\begin{lemma}\label{neighbors}
  \begin{enumerate}
\item 
  The $13$ potential neighbors of a typical vertex $A^L_i$
in $G$ are
all $5$ vertices of level $i$
\textup(%
$B^L_{i},C^L_{i};$
$A^R_i,B^R_{i},C^R_{i}$\textup),
all but one vertex of level $i-1$
\textup(%
$A^L_{i-1},B^L_{i-1};$
$A^R_{i-1},B^R_{i-1},C^R_{i-1}$
\textup)
 plus the $3$ vertices
$A^R_{i-2};$
$A^L_{i+1},C^L_{i+1}.$
\item
  The $15$ potential neighbors of a typical vertex $A^R_i$
in $G$ are
all $11$ vertices of levels $i$ and $i+1$
\textup(%
$A^L_{i}, B^L_{i},C^L_{i};$
$B^R_{i},C^R_{i};$
$A^L_{i+1},B^L_{i+1},C^L_{i+1};$
$A^R_{i+1},B^R_{i+1},C^R_{i+1}$~\textup)
plus the $4$ vertices
$A^R_{i-1},B^R_{i-1},C^R_{i-1};$
$A^L_{i+2}.$
\end{enumerate}
\end{lemma}

The claim immediately follows from the next
two lemmas.

\begin{figure}[tb]
  \begin{subfigure}[b]{.48\columnwidth}
  \centering
  \includegraphics[scale=.85]{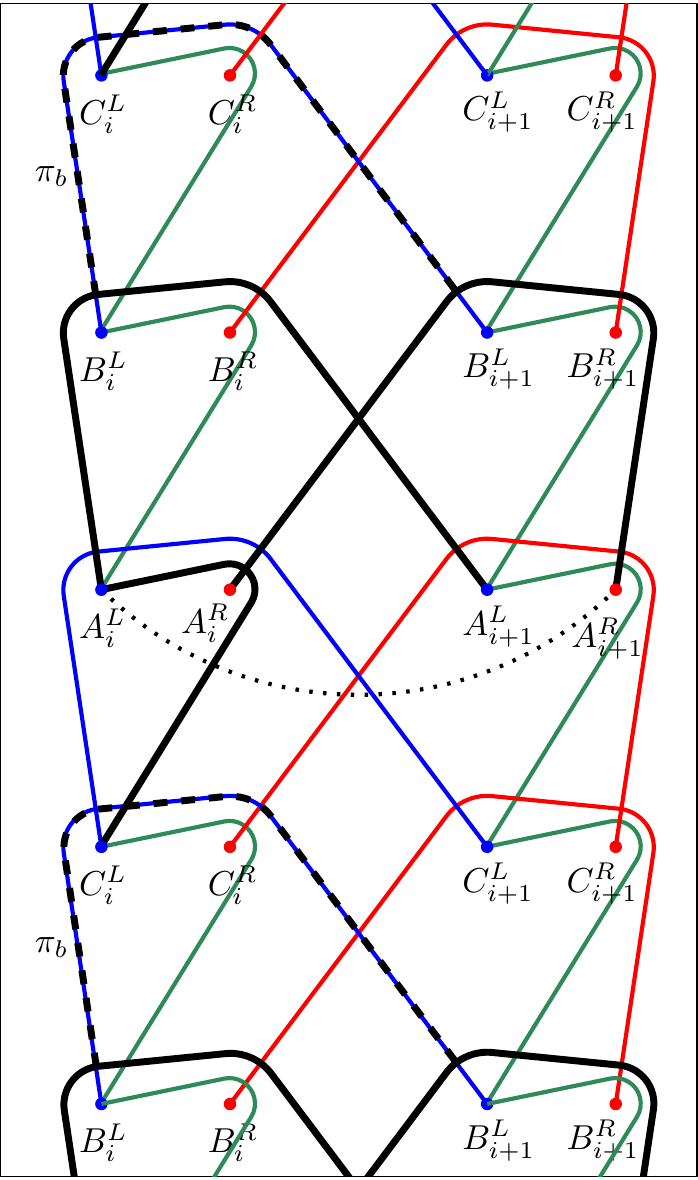}
  \caption{$A_i^L$ and $A^R_{i+1}$ cannot be connected.}
  \label{straight-right}
  \end{subfigure}
  \hfill
  \begin{subfigure}[b]{.48\columnwidth}
      \centering
  \includegraphics[scale=.85]{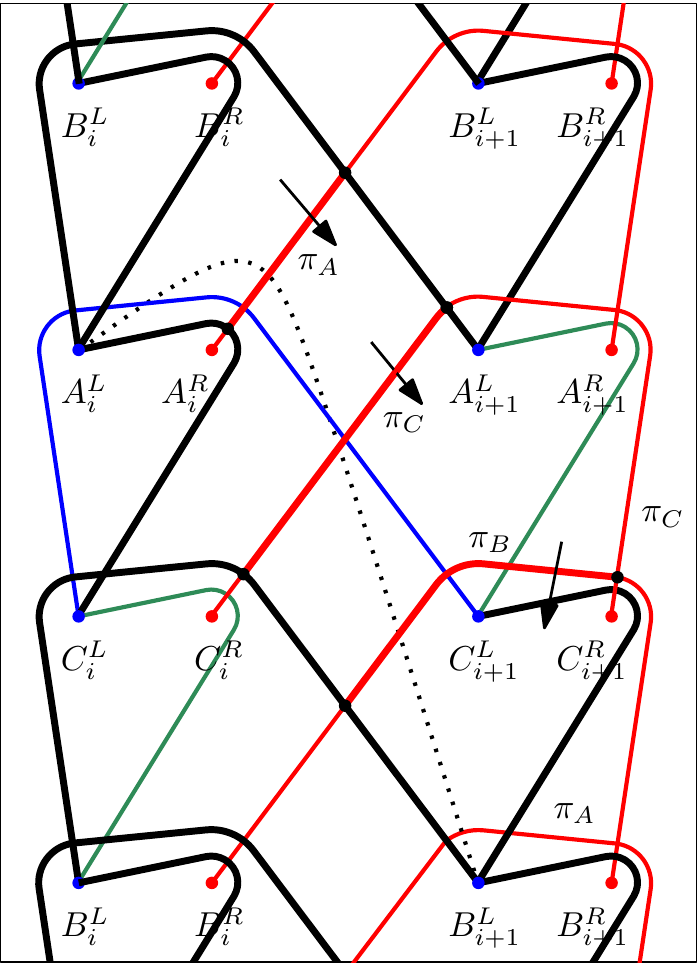}
  \caption{$A_i^L$ and $B^L_{i+1}$ cannot be connected.}
  \label{diagonal}
  \end{subfigure}
  \caption{Restricting more neighbors by putting back the green edges.}
  
\end{figure}

\begin{lemma}
  In a simple extension  of $G$, $A^R_{i+1}$ cannot be a neighbor
of $A_i^L$.
\end{lemma}
\begin{proof}
We call the edges incident to $A^R_{i+1}$ and $A_i^L$ the \emph{closed edges}.
Let $\pi_b$ the portion of the edge connecting $B_i^L$ with $B_{i+1}^L$ 
that runs above the closed edges (see Fig.~\ref{straight-right}).
The cell ``below'' $\pi_b$ is only bounded by $\pi_b$ and the closed edges. Hence, we cannot leave this cell
once we have entered. As a consequence we cannot cross $\pi_b$. Since the closed edges
together with $\pi_b$ disconnect  $A^R_{i+1}$ and $A_i^L$, these two vertices cannot be neighbors.
\end{proof}

\begin{lemma}
  In a simple extension  of $G$, $B^L_{i+1}$ cannot be a neighbor
of $A_i^L$.
\end{lemma}

\begin{proof}
All edges that are incident to either $B^L_{i+1}$ or $A_i^L$
cannot be crossed. These edges are drawn as black curves
in Fig.~\ref{diagonal} and are now considered as being the \emph{closed} edges.
The only chance to connect  $A_i^L$ with $B^L_{i+1}$
is to enter the region that is bounded by the closed edges and the
edge $\pi_A$ from $A_i^R$ to $A_{i+1}^R$. Thus we have to cross
this edge to leave this face. This leads us to a region that is bounded
by the closed edges, $\pi_A$ and the edge $\pi_C$ from $C_i^R$ to $C_{i+1}^R$.
Clearly we have to cross $\pi_C$ to leave this region. Now we have 
entered a region that is bounded by a closed edge, $\pi_C$ and
the edge $\pi_B$ that connects $B_i^R$ with $B_{i+1}^R$. To leave
this region we have to cross $\pi_B$, which brings us to a region that 
is bounded by a closed edge, $\pi_A$ and $\pi_C$. We observe that 
we are stuck in this region and hence, cannot reach $B^L_{i+1}$.
\end{proof}

By symmetry,
 $C^L_{i-1}$ and  $A_i^L$ cannot be neighbors, and 
this concludes the proof of Lemma~\ref{neighbors}.
Moreover, as a consequence of Lemma~\ref{neighbors}  the average degree
in a saturated extension of $G$ is at most $14$, which proves
Theorem~\ref{SimpleTheorem} when the number $n$ of vertices is a
multiple of~6.

We can determine the vertex degrees more carefully.
If $k\ge 3$,
then
\begin{enumerate}
\item 
the degrees of $A_1^L,B_1^L,C_1^L$  are at most $7$,
\item 
the degrees of $A_1^R,B_1^R,C_1^R$  are at most $12$,
\item 
the degrees of $A_2^L,B_2^L,C_2^L$  are at most $12$,
\item 
the degrees of $A_i^R,B_i^R,C_i^R$  are at most $15$, when
$1<i<k-1$,
\item 
the degrees of $A_i^L,B_i^L,C_i^L$  are at most $13$, when
$2<i<k$,
\item 
the degrees of $A_{k-1}^R,B_{k-1}^R,C_{k-1}^R$  are at most $14$,
\item 
the degrees of $A_k^L,B_k^L,C_k^L$  are at most $11$,
\item 
the degrees of $A_k^R,B_k^R,C_k^R$  are at most $8$.
\end{enumerate}
A straightforward calculation gives that
any
saturated extension of $G$ has at most $7n-30$
edges.
For $k=2$, the degrees of $X_1^L,X_1^R,X_2^L,X_2^R$ are bounded by
$7,11,10,8$, respectively, for a total of $54$ edges, which also agrees
with the formula $7n-30$.
Hence, for any $n\ge 12$ that is a multiple of~6, there exists a
saturated simple topological graph with $n$ vertices and at most $7n-30$ edges.

Our construction can be extended 
to any vertex size by \emph{cloning} some vertices.
Take a saturated simple topological graph
and any vertex $P$ of it. Next to $P$ we add $\rho$
new copies of $P$ -- the clones.
Connect the neighbors of $P$ to each clone
by edges that are non-intersecting
perturbations of the edges incident to~$P$.
By this we obtain a simple drawing.
 A saturation of this drawing
can include as additional edges
only edges among $P$ and its clones.

For $n\ge 12$,
we can write $n$ as $6k+\rho$ where $0 \le \rho \le 5$.
If $\rho=0$, we are done.
If $\rho\ge 1$, then start with a construction for a saturated simple
topological graph with $6k$ vertices.
Add $\rho$ clones of its lowest-degree vertex $P$, and saturate.
In our construction, the lowest degree is 7. Cloning such a vertex
$\rho$ times adds up to $7\rho+\binom{\rho+1} 2$ additional edges after saturation.
Since $\rho\le 5$,
the number of edges is bounded by
$$
7(6k)-30 + 7\rho+ {\textstyle \binom{\rho+1}2}
\le
7(6k+\rho)-30 + 15
=
7(6k+\rho)-15 < 7n
$$
The resulting simple topological graph
proves Theorem~\ref{SimpleTheorem} for $n\ge12$.
If $n\leq 11$, then the bound of Theorem~\ref{SimpleTheorem} holds since even the complete graph has at most
$\binom n2 \le 5n$ edges.

%
%
%
%
%








\section{Saturated 2-simple topological  graphs with few edges}
\label{sec:2simple}
In this section we construct sparse saturated drawings in the 2-simple case. 
We first review an auxiliary structure called grid-block. 
Then we use it to construct an efficient blocking edge configuration. 
We finish with explicit constructions of sparse saturated 2-simple drawings.
 
\subsection{The grid-block configuration} 
To begin, we study a drawing of 6 edges (three red edges and three black edges)  
as depicted in Fig.~\ref{fig:32gridfaces}. 
The drawing consists of three disjoint horizontal segments 
representing the red edges $r_1$, $r_2$, $r_3$, 
and three disjoint black edges $b_1$, $b_2$, $b_3$ that are drawn such that 
one crosses (in order) $r_1,r_2,r_3,r_1,r_2,r_3$, the other
$r_2,r_3,r_1,r_2,r_3,r_1$, and the last one $r_3,r_1,r_2,r_3,r_1,r_2$. There
are no other crossings in the drawing. Note that the configuration superimposes 
a grid. We call such an arrangement of edges a  {\emph {grid-block}}. These 
blocks have been used by Kyn\v{c}l~et~al.\ as building blocks  in their saturated graphs~\cite{kprt}. In the terminology 
of Kyn\v{c}l~et~al., our grid-blocks are named (3,2)-grid-blocks. 

As done in the previous section we consider the graph as drawn on the cylinder.
More precisely, we draw the graph inside a rectangle 
in which we identify two sides
in opposition  (\emph{bottom side} and \emph{top side}), 
while the other sides are named \emph{right side} and \emph{left side}. 
If an edge uses the transition across the bottom/top edge 
we say that it \emph{wraps around}.
In the following we assume that the grid-blocks are drawn 
such that only the black edges wrap around.
We label every face of the drawing of a grid-block with 2 numbers. 
These numbers refer to the coordinates of the 
(dual) superimposed grid, with $(0,0)$ being the label of the 
face that contains the two bottom most endpoints of the black
edges on the left side. All ``vertical'' coordinates are considered modulo~3.

\begin{figure}[htb]
\centering
\includegraphics{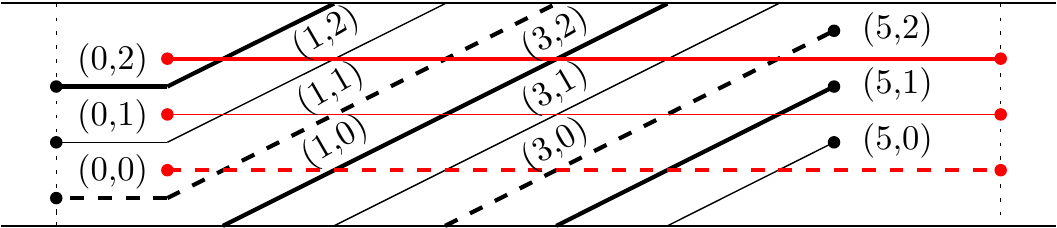}
\caption{A  grid-block with some labeled faces.}
\label{fig:32gridfaces}
\end{figure}

Throughout the section we study {\em paths} connecting the left and the right sides of the cylinder and passing through some blocking configurations. By a {\em path} in this context we always mean a path in a graph dual to the arrangement of the blocking configuration in question.

Kyn\v{c}l et al.\ observed that every path connecting the left with the right side of the cylinder
has to intersect the edges of grid-block at least 5~times. For our construction we
need a stronger statement which is presented in Lemma~\ref{lm:simplification}. The following 
lemma simplifies the treatment of paths passing through the grid-block.

\begin{lemma}
\label{lm:simplification}
Let $\gamma$ be a path crossing the grid-block that starts
in face $(0,i)$ and ends in face $(5,j)$ and that 
never visits the faces $(0,\cdot)$, $(5,\cdot)$ again, 
see Fig.~\ref{fig:32gridwithpath}.
Then $\gamma$ can be transformed, keeping its endpoints fixed, 
to a path $\widetilde\gamma$ 
such that $\widetilde\gamma$:
\begin{enumerate}
\item crosses (with the same or smaller multiplicity) only the edges of the grid-block 
crossed by $\gamma$,
\item first walks between the faces $(0,i)$, ${0\le i\le 2}$,
then crosses some black edges to the right, 
passing from a face~$(k,i)$ to a face~$(k+1,i)$,
then crosses some red edges upwards, 
passing from a face~$(k,i)$ to a face~$(k+1,i+1)$.
\end{enumerate} 
\end{lemma}

\begin{figure}[htb]
\centering
\includegraphics[scale=1.2]{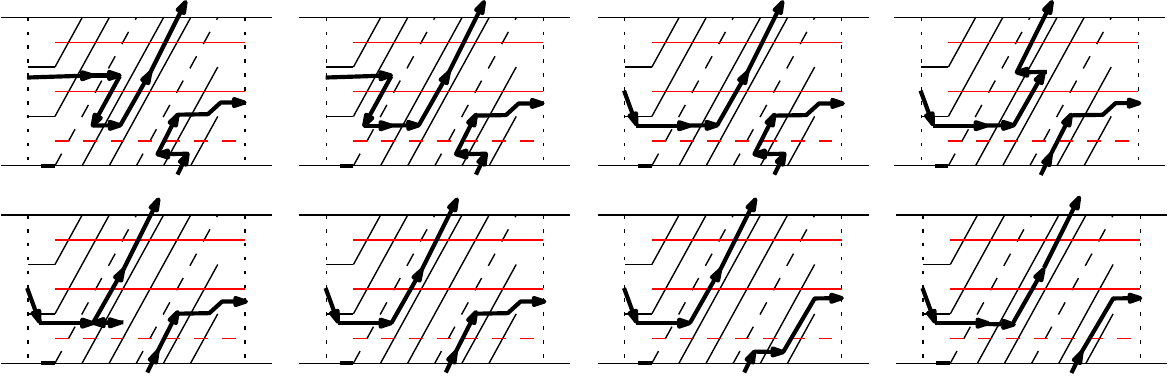}
\caption{Process of the simplification of a path passing through a grid-block.}
\label{fig:32gridwithpath}
\end{figure}

\begin{proof}
We refer to the transition of the path 
from one cell of the arrangement to an adjacent cell as a
\emph{step}.
There are
four different types of steps:
$\rightarrow$, $\leftarrow$, $\nearrow$ or $\swarrow$, depending on
the crossed edge and the direction, see
Fig.~\ref{fig:32gridwithpath}.

We execute the path simplification through a series of local 
modifications on pairs of two consecutive steps: (1) annihilation of two
consecutive steps in opposite directions and (2) changing places of two
consecutive steps that are not yet in a desired order.

The simplification is carried out in two stages. 
 In the first stage (shown in the first 6
pictures in Fig.~\ref{fig:32gridwithpath}) we remove all 
``backward steps'' $\swarrow$ and $\leftarrow$, while possibly
increasing the number of steps the path $\widetilde\gamma$ walks
between faces $(0,i)$, $0 \le i\le 2$.  In the second stage we reorder the
steps $\nearrow$ and $\rightarrow$ such that no $\nearrow$ precedes any
$\rightarrow$.
\begin{figure}[htb]
\centering
\includegraphics{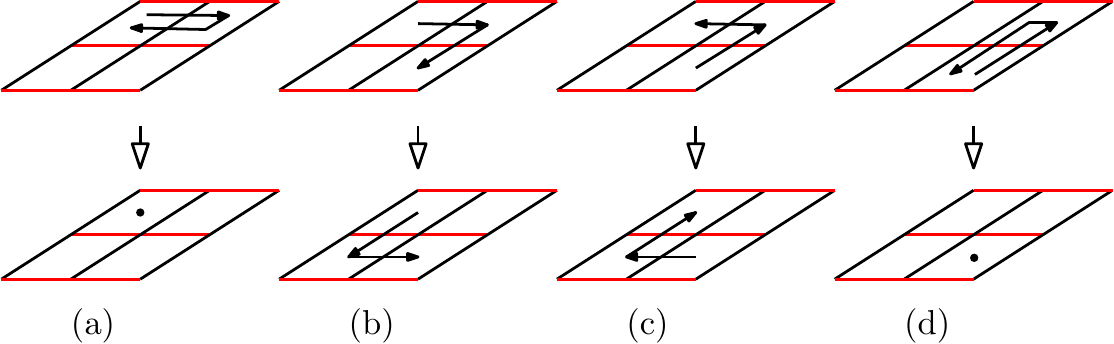}
\caption{4 possible 2-step configurations involving ``backward steps''
  as a second step before (first row) and after (second row) the
  appropriate local modification.}
\label{fig:32gridforbiddenconf}
\end{figure}

{\bf Stage 1:}
We traverse the path until we meet the first $\leftarrow$ or $\swarrow$
step. Together with its preceding step it forms one of the 4
configurations shown in~Fig.~\ref{fig:32gridforbiddenconf}.  In cases
(a) and (d) the steps only differ in their orientation, hence we can annihilate two steps.
In the remaining cases (b) and (c) we reorder the two steps.
This reordering can be safely executed unless it forces the path to
leave the grid-block. 
This, however, may happen only when the
backward step ($\leftarrow$ or $\swarrow$) starts from one of the faces labeled $(2,k)$.  Since this
backward step is the first backward step of the path, we are left with two subcases for each (b) and (c)
depending on the preceding step, which might be either $\rightarrow$ or $\nearrow$.
The four cases are depicted in
Fig.~\ref{fig:32gridbadcase}. All the cases can be handled by further local 
simplifications that are shown in the figure.
\begin{figure}[htb]
\centering
\includegraphics[scale=0.9]{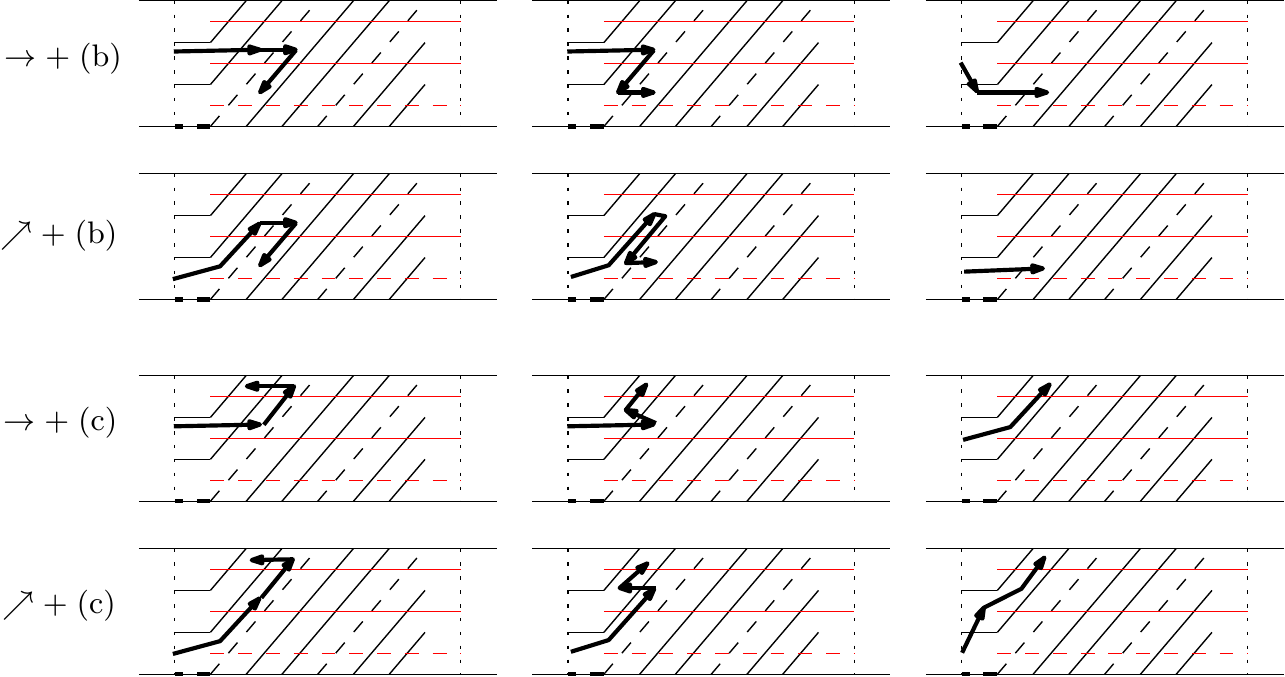}
\caption{Handling of the 4 possible cases (each row represents one case) when
  the local modifications (b) or (c) force the path out of the grid-block.}
\label{fig:32gridbadcase}
\end{figure}

We finish the proof of the stage 1 using double induction on the
number of backward steps and, within it, on the distance from the
beginning of the path to the first backward step.

{\bf Stage 2:}
After the stage 1 our path through the grid,
leaving aside its first steps between faces $(0,i)$,
has only $\rightarrow$ and $\nearrow$ steps. 
These two types of steps can be reordered 
without changing the number of times 
the path crosses any edge of the grid. 
Moreover this reordering never leads the path out of the
grid-block.
\end{proof}


\subsection{A blocking configuration}

We call the building blocks of the following constructions 
\emph{black block} and \emph{red block}, see Fig.~\ref{fig:blackredblocks}. 
We refer to the edges of the red (black) block as 
\emph{red edges} (\emph{black edges}).
Any two red edges, as well as any two black edges, cross exactly twice. 
Note that up to a reflection the red block is homeomorphic to the black block.
\begin{figure}[htb]
\centering 
\includegraphics{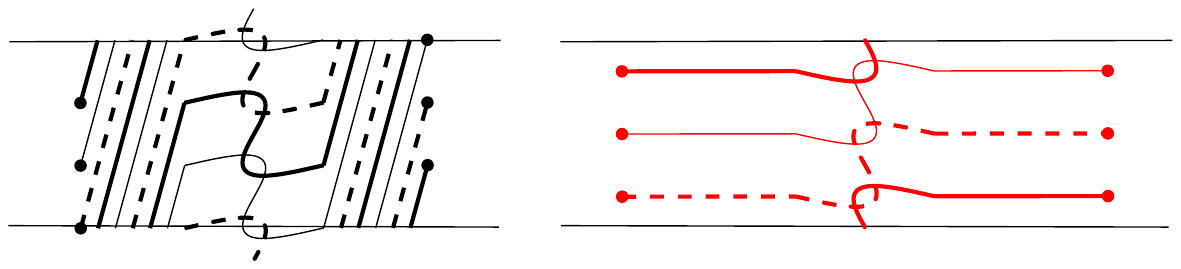}
\caption{A black (left) and a red (right) blocks.}
\label{fig:blackredblocks}
\end{figure}

We combine two black blocks and a red block as shown in Fig.~\ref{fig:3block} 
to obtain a drawing that we call a {\emph {3-block}}. 
Since the red block differs from the black block 
only by a reflection, the 3-block built form
consecutive black-red-black blocks is a mirror image of the 3-block
built from consecutive red-black-red blocks.
\begin{figure}[htb]
\centering
\includegraphics[scale=1.2]{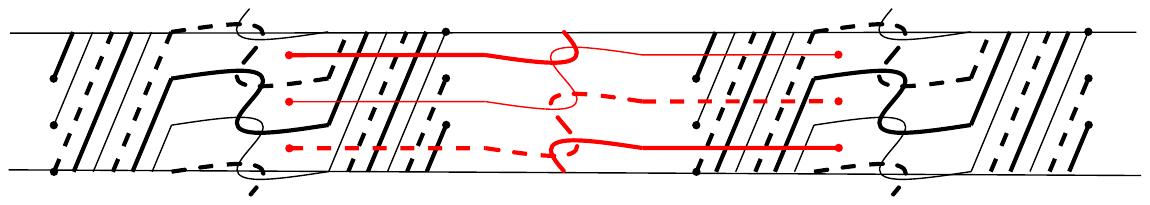}
\caption{A 3-block, formed by consecutive black, red and again black blocks.}
\label{fig:3block}
\end{figure}

The following theorem is the key observation that we need for the construction of the sparse 2-simple drawing.
\begin{theorem}
\label{thm:blocking}
Any path connecting the left with the right sides of the cylinder while passing through the 3-block
 crosses one of the edges forming the 3-block at least 3 times.
\end{theorem}
Before proving the theorem we provide some helpful lemmas.
We label some of the faces of the arrangement as 
shown in Fig.~\ref{fig:3blockeverything}. In particular, for $i=0,1,2$, we denote the faces containing
the left endpoint of the red edges $\rl_i$ as $\LZ_i$, and the faces containing the right endpoint as $\RZ_i$.
The edges of the left black block are named $\bl_i$ and the edges of the right black block are 
named $\bl'_i$. Finally, let $\LMZ_i$ be the face  that contains the right endpoint of $\bl_i$, and let 
$\RMZ_i$ be the face  that contains the left endpoint of $\bl'_i$.
The region spanned by $\LZ_0$, $\LZ_1$ and $\LZ_2$ is denoted by $\LZ$.
We similarly define regions $\LMZ$, $\RMZ$ and $\RZ$.

\begin{figure}[htb]
\centering
\includegraphics{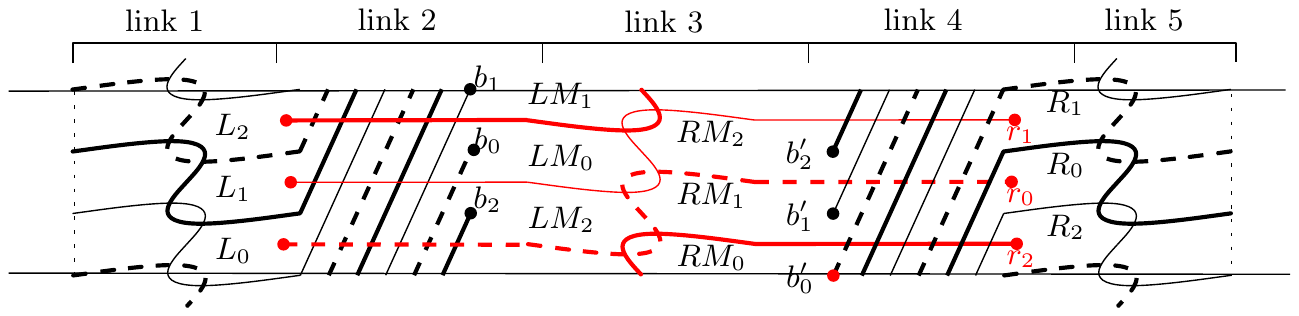}
\caption{A 3-block with some distinguished faces (capital letters) and edges.
 The red edges forming the blocks are labeled $\bl_i$, $\bl'_i$ and $\rl_i$. 
 The ``zones'' at which we subdivide the path into links are labeled above the strip.}
\label{fig:3blockeverything}
\end{figure}

Let $\gamma$ be a path that passes the 3-block.
To facilitate the analysis we subdivide the path $\gamma$ into smaller pieces, which we call \emph{links}. The 
links are defined as follows:
\begin{enumerate}[leftmargin=1.8cm,itemsep=0ex]
\item[link 1:]  from the start point (left) of $\gamma$ to the last point of $\gamma$ in  $\LZ$,
\item[link 2:]  from the last point of $\gamma$ in $\LZ$ to its first point in $\LMZ$,
\item[link 3:]  from the first point of $\gamma$ in $\LMZ$ to its last point in $\RMZ$,
\item[link 4:]  from the last point of $\gamma$ in $\RMZ$ to its first point in $\RZ$,
\item[link 5:] from the first point of $\gamma$  in $\RZ$  to its (right) endpoint.
\end{enumerate}
Before we proceed we check that the links are well defined, i.e., that
the points defining the links appear in order.
For the links~1, 3 and 5 this holds  trivially, while to check it for links 2
(and, symmetric, 4), we need to prove that the last point in $\LZ$ precedes
the first point in $\LMZ$:

\begin{lemma}
No path can visit the regions $\LZ \rightarrow \LMZ \rightarrow \LZ \rightarrow \LMZ$ in this order 
without crossing some of the edges forming the 3-block at least 3 times.
\end{lemma}
\begin{proof}
The faces $\LZ$ and $\LMZ$ are separated by a grid-block.  Passing
through it requires at least 5 crossings of its edges.  Any
path visiting $\LZ$ $\rightarrow$ $\LMZ$ $\rightarrow$ $\LZ$
$\rightarrow$ $\LMZ$ would cross the grid-block at least 3
times, and hence it would cross the edges of the grid-block at least $3\times 5 = 15$ times.
Since a grid-block is formed by 6 edges, at least one of
them will be crossed 3 times or more.
\end{proof}

We continue by analyzing the path through the 3-block following its links.
\begin{lemma}
\label{lm:step1}
Any path passing the 3-block from left to right with the last point of link~1 at $\LZ_i$
crosses the edge $\bl_{i+1}$ at least once or one of the edges $\bl_{i}$ and
$\bl_{i+2}$ at least twice at its first link (all indices modulo 3).
\end{lemma}
\begin{proof}
A path that ends in $\LZ_i$ crosses either $\bl_{i+1}$ or it crosses $\bl_{i+2}$ while entering from $\LZ_{i+1}$.
Repeating this argument twice proves the lemma.
\end{proof}

%
The following lemma summarizes the behavior of the path on the first two links:
\begin{lemma}
\label{lm:modification1}
Any path $\gamma$ passing the 3-block that does not intersect 
any edge 3 times or more crosses the 
red edges $\rl_{j}$, $\rl_{j+1}$ before it first visits 
the region $\LMZ$ at $\LMZ_j$.
\end{lemma}
\begin{proof}
We modify the path $\gamma$ along link~2 following the
simplification procedure described in Lemma~\ref{lm:simplification} to
get a path $\widetilde\gamma$.
Lemma~\ref{lm:simplification} also implies that the link~2 of
$\widetilde\gamma$ consists of exactly~5 ``steps'': first, $0\le h\le 5$
steps crossing the black edges $\rightarrow$ to the right, followed by
$v=5-h$ steps crossing red edges $\nearrow$ upward.

Assume that the first point of link~2 of
$\widetilde\gamma$ lies inside the face $\LZ_{i}$.  Then~$h$ horizontal steps of link~2 
cross the  $\bl_{i+1}$, $\bl_{i}$, $\bl_{i-1}$,
\ldots, $\bl_{i+1-(h-1)}$.  Moreover, Lemma~\ref{lm:step1} guarantees that
already link~1 of the path $\widetilde\gamma$ crossed either
$\bl_{i+1}$ once or one of $\bl_{i}$ or $\bl_{i+2}$ twice.  Since $\widetilde\gamma$
does not cross any of the black edges more than twice, it follows that $h\le 3$.
This, however, shows that $v\ge 2$, which implies that the path $\widetilde\gamma$ crosses
the red edges $\rl_{j+1}$, $\rl_{j}$ before it reaches the last point of
its second link in face $\LMZ_{j}$.
To finish the proof we recall that the path $\gamma$ crosses every
edge of the 3-block at least as many times as $\widetilde\gamma$
and that the last points of the link 2 of $\gamma$ and $\widetilde\gamma$ coincide.
\end{proof}

\paragraph{\it Proof of Theorem~\ref{thm:blocking}.}
We prove by contradiction, namely, we assume that there is a path $\gamma$
that passes through the 3-block while crossing every edge of the 3-block at most twice.
Let~$\LMZ_{j}$ be the face where link~2 ends, and let~$\RMZ_{k}$ be the face where 
link~4 starts. By Lemma~\ref{lm:modification1} we know that $\gamma$ crosses 
 $\rl_j$ and $\rl_{j+1}$ in link~1 and link~2. Since the structure of the link~4 and~5 coincides with the structure
of link~2 and~1  we can apply Lemma~\ref{lm:modification1} also to the last 2 links.
Thus, $\gamma$ crosses $\rl_{k-1}$, $\rl_{k}$ in link~4 and~5.
A short case distinction ($k$ might be either $j$, $j+1$, or $j+2$) shows that $\gamma$ cannot
connect endpoints of link~2 and~4 via link~3
without crossing at least one of the red edges 3 times; see
Fig.~\ref{fig:laststep}.  The figure depicts all 
 ways of how to possibly route the path $\gamma$ in link~3.  
Each of the possible continuations crosses some of the
red edges $\rl_j$, $\rl_{j+1}$, $\rl_{j-1}$ twice and is blocked
within one of the faces before it reaches the face $\RMZ_{k}$.
As a consequence the path $\gamma$ cannot exists.
\qed

\begin{figure}[htb]
\centering
\includegraphics[scale=1.2]{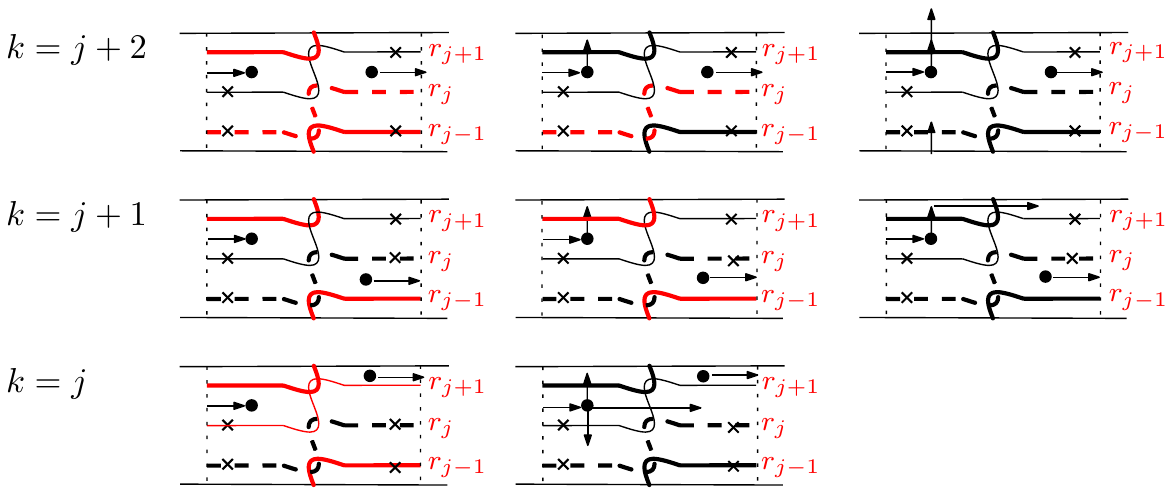}
\caption{Each row depicts a case.  Black
  dots inside faces mark the faces $\LMZ_{j}$ (left) and $\RMZ_{k}$
  (right). Black crosses on red edges mark the edges that are,
  due to Lemma~\ref{lm:modification1}, crossed by the path outside
  link~3.  We color red edges black as soon as they are crossed
  by the path $\gamma$ twice and no more crossings are allowed.  In the case
  $k=j$ the path can be continued in 3 different direction, in each of
  them the path is blocked after one step.}
\label{fig:laststep}
\end{figure}

\subsection{A sparse saturated 2-simple drawing}

We show next how to combine a sequence of 3-blocks in order to obtain a 2-simple saturated drawing with few edges
to obtain the following result.
%
\begin{theorem}
Let $s_2(n)$ denote the minimum number of edges that
a 2-simple saturated drawing with $n$ vertices can have.
Then
$s_2(n)\le 14.5 n.$
\end{theorem}


\begin{proof}
We consider the drawing that repeats the pattern shown in Fig.~\ref{fig:combined}.
The horizontal strip denotes the cylinder.  The drawing is formed by
$k$ consecutive black and red blocks; see
Fig.~\ref{fig:blackredblocks}.  Each block contains 6 vertices, so the
total number of vertices is $k\times 6$.  Clearly, the drawing is
2-simple.
\begin{figure}[htb]
\centering
\includegraphics{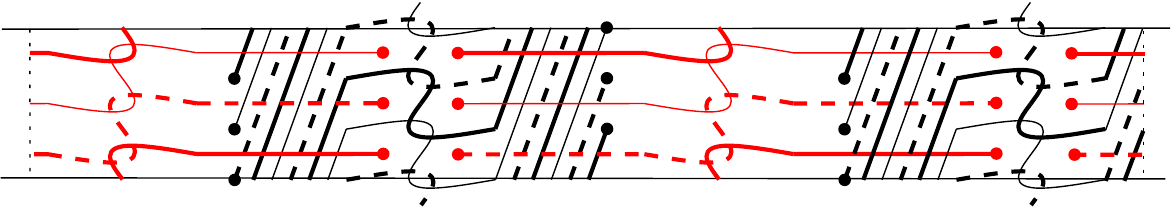}
\caption{A 2-simple drawing that does not allow too many edges to be added.}
\label{fig:combined}
\end{figure}

Now we add as many edges as possible without violating the
2-simplicity, so that the drawing becomes saturated (this padding
procedure is definitely not unique).  Theorem~\ref{thm:blocking}
implies that without violating the 2-simplicity any vertex can be
connected by an edge only to~29 other vertices; see
Fig.~\ref{fig:connection} for ``internal'' vertices and
Fig.~\ref{fig:connectionspecial} for vertices close to the left (right)
boundary of the cylinder. This implies that the maximal number of
edges in the resulting saturated 2-simple drawing is less or equal
than $14.5n$.
\begin{figure}[htb]
\centering
\includegraphics[scale=1.07]{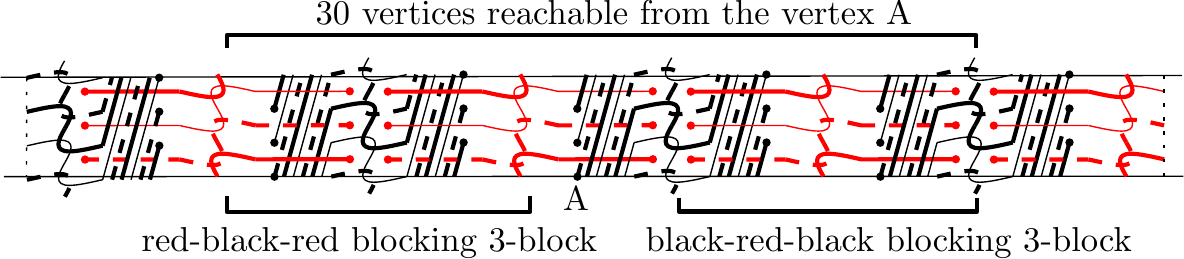}
\caption{The potential neighbors of a typical vertex $A$.}
\label{fig:connection}
\end{figure}

For $n$ not divisible by~$6$ we build the construction above with
$k=\lfloor n/6\rfloor$.  We split the remaining $l=n-6\lfloor n/6
\rfloor$ vertices into two groups of no more than 3
vertices each, and place one group with $l_1$ 
vertices to the left and one group with $l_2$ vertices on
 to the right of the resulting arrangement.

\begin{figure}[htb]
\centering
\includegraphics[scale=1.13]{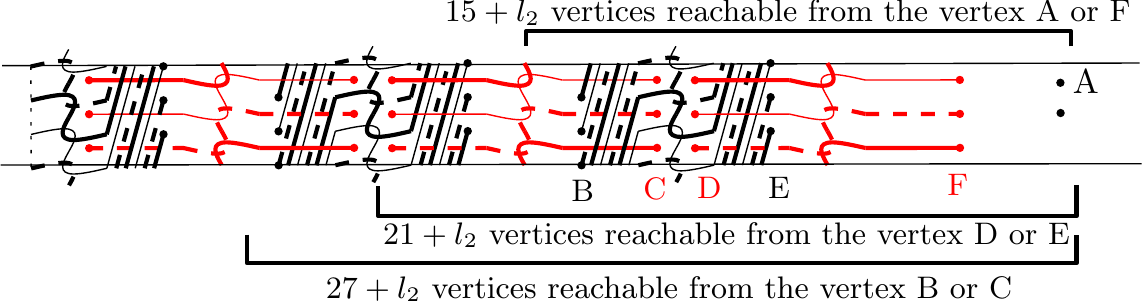}
\caption{The potential neighbors of vertices close to the boundary.}
\label{fig:connectionspecial}
\end{figure}

The possible connections with the newly introduced vertices are
illustrated in Fig.~\ref{fig:connectionspecial}.  Since $l_1,
l_2\le 3$, no vertex has degree greater than $29$.
\end{proof}


\section{Local saturation}
\label{sec:local}
\subsection{Simple drawings}
The lower bound in \cite{kprt} on the number of
edges in a
saturated simple topological graph is based on the following
lemma.

\begin{lemma}[\cite{kprt}]
Let $G$ be a simple topological graph with at least four vertices, and
let $A$ be
a vertex
 of degree at most two.
Then
$G$ has a simple extension by an edge incident to $A$.
\end{lemma}

This lemma implies that in a simple saturated
topological graph with at least four vertices,
every vertex
must have degree at least three, and hence the number of edges is at
least $1.5n$.
Can we improve the bound on the edge number by strengthening the lower
bound on the degree? The following considerations establish a limit to
this approach: There are saturated graphs with minimum degree four.

We say that a vertex $S$
in a simple topological graph
 is \emph{
saturated} if
it cannot be connected to a non-adjacent vertex
while maintaining simplicity.
The above lemma implies that in a simple 
topological graph with at least four vertices,
 a saturated vertex
must have degree at least three.

\begin{observation}
For any positive integer $n\ge6$,
there is a simple topological graph
on $n$ vertices with a saturated vertex of degree four. 
\end{observation}

\begin{figure}[htb]
  \centering
  \includegraphics{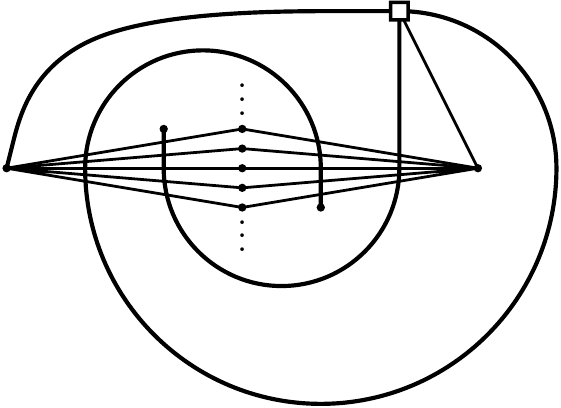}
  \caption{The boxy vertex of degree four is saturated.}
  \label{Minimum_degree_4}
\end{figure}

The observation is due to the construction presented in Fig.~\ref{Minimum_degree_4}.
This example is an extension of the case $n=6$ from
\cite[Fig. 2]{kprt}.
The topmost vertex is saturated since all incident faces
are bounded by edges incident to that vertex.


\begin{figure}[htb]
  \centering
  \includegraphics{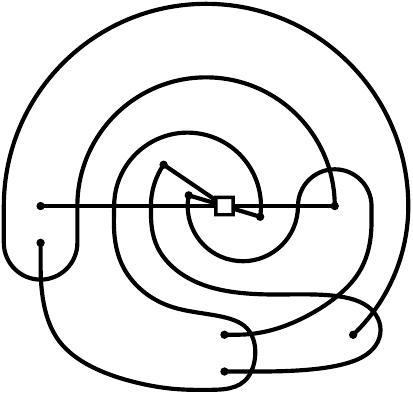}
  \caption{In the simple topological graph above, the central vertex
    has degree~5, and it
           cannot be connected by an edge to any point in 
           the unbounded region while keeping simplicity.}
  \label{mindegree5}
\end{figure}

The following lemma presents a construction that realizes
small vertex degrees for many vertices.
\begin{lemma}
For any positive integer $k$,
there exists a saturated simple topological graph on $10k$ vertices
with $k$ vertices of degree~$5$.
\end{lemma}
\begin{proof}
The main idea of our construction is depicted in Fig.~\ref{mindegree5}.
A simple case distinction verifies that no edge can connect the central vertex with a point
on the outer face without violating the simplicity of drawing. 

Now, take $k$~copies of the drawing in Fig.~\ref{mindegree5},
and place them on the plane next to each other such that the interior faces 
of the copies are non-overlapping.
The $k$~copies of the central vertex will remain degree-$5$ vertices
no matter how we saturate the graph.
\end{proof}

\subsection{2-simple drawings}

To study local saturation in 2-simple case we use a slight modification of the 3-block introduced in Sect.~\ref{sec:2simple}; see Fig.~\ref{fig:local2saturation}.

\begin{figure}[htb]
\centering
\includegraphics[scale=1.2]{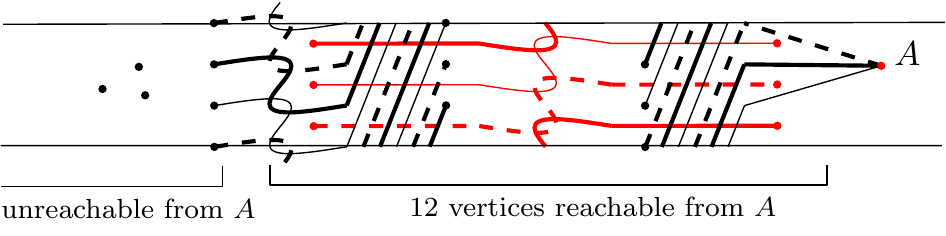}
\caption{The rightmost vertex~$A$ cannot be connected to any vertex that belongs to the leftmost (unbounded) face without violating 2-simplicity.}
\label{fig:local2saturation}
\end{figure}

By the arguments given in the proof of Theorem~\ref{thm:blocking} 
the rightmost vertex can be connected 
to only~$12$ other vertices (Fig.~\ref{fig:local2saturation})
and thus cannot be connected to any vertex that belongs to the leftmost (unbounded) face of the drawing
without violating 2-simplicity.

\begin{figure}[htb]
\centering
\includegraphics{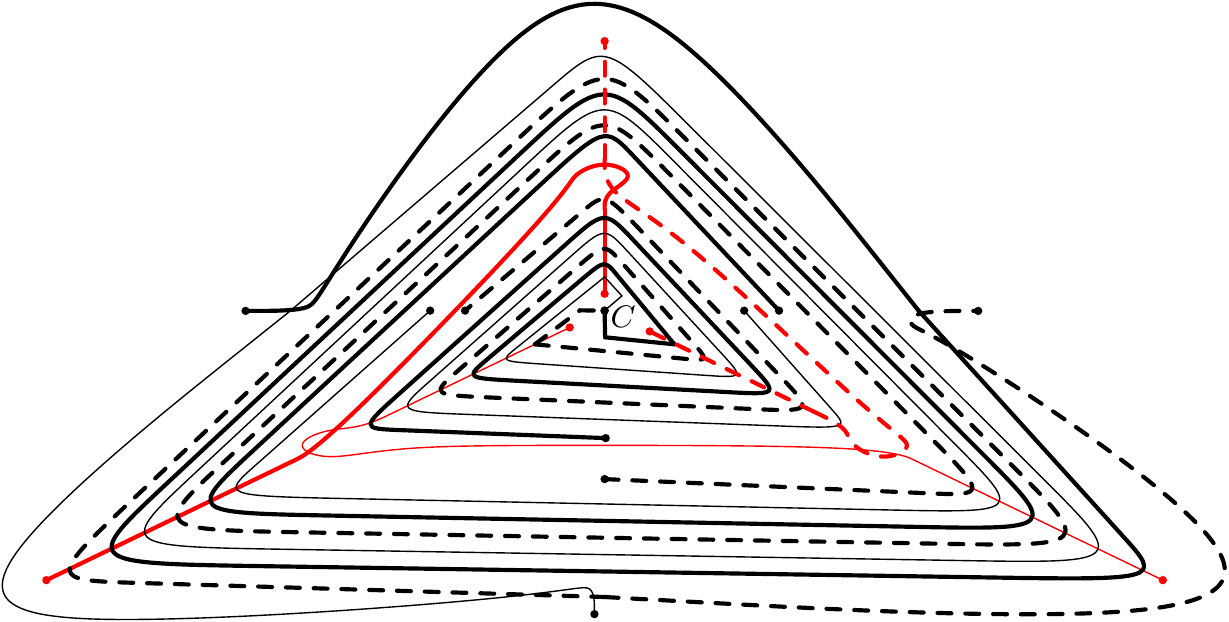}
\caption{Unrolling of Fig.~\ref{fig:local2saturation} to the plane. The central vertex~$C$ corresponds to the rightmost vertex~$A$ of Fig.~\ref{fig:local2saturation}.}
\label{fig:unrolled2loc}
\end{figure}

The ``unrolling'' of this configuration from the cylinder to the plane (with center of the unrolling in the rightmost vertex)
is presented on Fig.~\ref{fig:unrolled2loc}. 
The central vertex cannot be connected by an edge to any vertex that belongs to the unbounded region without violating 2-simplicity, and so it has degree no larger than 12 in any saturation.
After placing  $k$ disjoint copies of this construction to the plane next to each other we obtain the  following result:
\begin{lemma}
For any positive integer $k$,
there exists a saturated 2-simple topological graph on $16k$ vertices
with $k$ vertices of degree $12$.
\end{lemma}


\section*{Acknowledgment}
The first author thanks G\'eza T\'oth for presenting 
their inspiring results (\cite{kprt}) in Szeged and
for the encouragement during
his investigation.
This research was partially initiated at the EuroGIGA \emph{Workshop on Geometric Graphs (GGWeek '14)} in M\"unster, Germany,
in September 2014,
supported by the European Science Foundation (ESF) through the Collaborative Research Program
\emph{Graphs in Geometry and Algorithms} (EuroGIGA).
We would like to thank all participants for the inspiring discussions.


\bibliographystyle{abbrv} \bibliography{saturated}

\end{document}